\documentclass[11pt]{article}
\usepackage{a4wide}
\usepackage{amsmath,amsfonts,amssymb,latexsym,stmaryrd}
\usepackage[PostScript=dvips]{diagrams}
\usepackage{pstricks,pst-node,pst-tree}

\newtheorem{Th}{Theorem}[section]

\newtheorem{proposition}[Th]{Proposition}

\newtheorem{definition}[Th]{Definition}
\newtheorem{example}[Th]{Example}

\newenvironment{proof}{\textbf{Proof.}}{$\;\;\; \Box$}

\newarrow{Dotline}.....

\newcommand{\ie}{\textit{i.e.}~}
\newcommand{\Nat}{\mathbb{N}}
\newcommand{\BB}{\mathbb{B}}
\newcommand{\CC}{\mathcal{C}}
\newcommand{\HF}{\mathbb{H}}
\newcommand{\Fin}{\mathbb{F}}
\newcommand{\power}[1]{\mathcal{P}(#1 )}
\newcommand{\powerf}[1]{\mathcal{P}_{f}(#1 )}
\newcommand{\Powerf}{\mathcal{P}_{f}}
\newcommand{\powerfk}[1]{\mathcal{P}_{f}^{k}(#1 )}
\newcommand{\powerH}[1]{\mathcal{P}_{H}(#1 )}
\newcommand{\powerh}{\mathcal{P}_{H}}
\newcommand{\powerV}[1]{\mathcal{P}_{V}(#1 )}
\newcommand{\powerv}{\mathcal{P}_{V}}
\newcommand{\powerP}[1]{\mathcal{P}_{P}(#1 )}
\newcommand{\powerPO}[1]{\mathcal{P}_{P}^{0}(#1 )}

\newcommand{\vn}{\varnothing}
\newcommand{\Init}{\mathbf{0}}
\newcommand{\Term}{\mathbf{1}}
\newcommand{\Two}{\mathbf{2}}
\newcommand{\Max}[1]{\mathsf{Max}(#1 )}
\newcommand{\MT}[1]{\mathsf{MT}(#1 )}
\newcommand{\MA}[1]{\mathsf{MA}(#1 )}
\newcommand{\MAk}[1]{\mathsf{MA}^{k}(#1 )}
\newcommand{\Form}[1]{\mathcal{F}(#1 )}
\newcommand{\lift}[1]{#1_{\bot}}
\newcommand{\CUM}{\mathbf{CUMet}}
\newcommand{\Stone}{\mathbf{Stone}}
\newcommand{\Dom}{\mathbf{Dom}}
\newcommand{\Set}{\mathbf{Set}}
\newcommand{\Bool}{\mathbf{Bool}}
\newcommand{\SFP}{\mathbf{SFP}}
\newcommand{\KK}{\mathbf{K}}
\newcommand{\lsem}{\llbracket}
\newcommand{\rsem}{\rrbracket}

\title{A Cook's Tour of the Finitary Non-Well-Founded Sets}

\author{Samson Abramsky\\
Oxford University Computing Laboratory}
\date{}

\begin{document}
\maketitle

\begin{abstract}
We give multiple descriptions of a topological universe of finitary sets, which can be seen as a natural limit completion of the hereditarily finite sets. This universe is characterized as a metric completion of the hereditarily finite sets; as a Stone space arising as the solution of a functorial fixed-point equation involving the Vietoris construction; as the Stone dual of the free modal algebra; and as the subspace of maximal elements of a domain equation involving the Plotkin (or convex) powerdomain. These results illustrate the methods developed in the author's `Domain theory in logical form', and related literature, and have been taken up in recent work on topological coalgebras. 
The set-theoretic universe of finitary sets also supports an interesting form of set theory. It contains non-well founded sets and a universal set; and is closed under positive versions of the usual axioms of set theory.
\end{abstract}

\section{Some reminiscences, and an explanation}
It is a great pleasure to contribute this paper to a birthday volume for Dov.
Dov and I arrived at imperial College at around the same time, and soon he, Tom Maibaum and I were embarked on a joint project, the Handbook of Logic in Computer Science. We obtained a generous advance from Oxford University Press, and a grant from the Alvey Programme, which allowed us to develop the Handbook in a rather unique, interactive way. We held regular meetings at Cosener's House in Abingdon (a facility run by what was then the U.K. Science and Engineering Research Council), at which contributors would present their ideas and draft material for their chapters for discussion and criticism. Ideas for new chapters and the balance of the volumes were also discussed. Those were a remarkable series of meetings --- a veritable education in themselves. I must confess that during this long process, I did occasionally wonder if it would ever terminate \dots . But the record shows that five handsome volumes were produced \cite{Han}.
Moreover, I believe that the Handbook has proved to be a really valuable resource for students and researchers. It has been used as the basis for a number of summer schools. Many of the chapters have become standard references for their topics. In a field with rapidly changing fashions, most of the material has stood the test of time --- thus far at least!

A large part of this success is due to Dov. Even though this particular Handbook series (among the many he has edited) is not the closest to his own interests, he not only originally inspired the project and got it going, but he stayed with it, and his energy and enthusiasm were essential to carrying it through. The ideas he had learned from his previous experience with the Handbook of Philosophical Logic \cite{HPL} proved important. For example, every Chapter had an official Second Reader, a friendly critic and conscience; in many cases, these Second Readers worked above and beyond the call of duty, and helped to materially improve the Chapters. Another of Dov's ideas was that each Chapter should have a broad division into three parts. A first part should be genuinely introductory, and provide a helpful overview to the browser --- who might then return for a more detailed look. The second part should be the technical core of the topic --- and the contents of this core should be agreed by a consensus, in the discussions at the Handbook meetings. Finally, in the last part the author was free to ride their own hobby-horses, and pursue those topics they were particularly keen on in greater depth. The wisdom behind this is that the freedom offered by this third part made accepting a wider consensus on the core of the Chapter much more palatable to authors.

Beyond these organizational ideas, Dov's presence at the Handbook meetings was crucial to establishing their distinctive, intellectually engaged but friendly and relaxed atmosphere. He infused these occasions with his inimitable sense of humour, and his vision of the great possibilities of cooperation in Science. 

So, thank you Dov, for this and much else --- but above all for drawing me into what proved to be such a worthwhile project.

\subsection*{An Explanation}
The scientific paper which follows requires a few words of explanation.
Recalling those days when we were engaged with the Handbook, roughly the period 1985--90, the thought occurred to me that writing up a lecture which I had given then, but never published, might be rather appropriate. For one thing, the lecture has something to say about modal logic, and hence is closer to Dov's interests than much of my work.

To whet the appetite of any modal logicians who may be reading this, let me challenge them with the following questions:
\begin{center} \textsf{What is the Stone space of the free modal algebra?\\
Which kind of  set theory does it provide a model for?}
\end{center}
(If you can answer these questions, you are excused from reading this paper.)

Lectures on versions of this material were given on several occasions in 1988-89, including:
the 1988 British Colloquium on Theoretical Computer Science in Edinburgh;
the Workshop on Logic from Computer Science held at MSRI Berkeley in 1989; and the International Symposium on Topology held in Oxford in July 1989.
The lecture has been referenced in several subsequent publications, e.g. \cite{OMM,BM,ABH,FH1}.
It has always been rather on my conscience that I had not written it up for publication.
That is what I have now done. I have added a few references to later work, and tidied up one or two points of technical detail, but essentially this is a straightforward write-up of the 1988 lecture. It is an extended discussion of a single example, which is used to illustrate some wider themes. It may not be \emph{new}, but I hope that it can still be \emph{useful}. In that sense, it is offered in the same spirit as the Handbook, which Dov and Tom and I were working on in those years.

\section{Introduction}
Our topic in this paper is a single example: the space of \emph{finitary non-well-founded sets}, which we study from many different points of view: process models; metric approximation; topology and the Vietoris construction; modal logic and Stone duality; and domain theory. We obtain five distinct characterizations of this space. We also study some basic features of its behaviour as a set-theoretic universe.

The main purpose of this study is to illustrate some general themes, in particular:
\begin{itemize}
\item Alternative descriptions of models:
\begin{itemize}
\item relating different semantics, and semantics and logic
\item deriving one systematically from the other.
\end{itemize}

\item Taking our cue from Domain theory \cite{Sco70}, the right level for mathematical modelling of computation is
\begin{itemize}
\item not the strictly finite
\item not the unboundedly infinite
\item but the \emph{finitary} \ie those objects appearing as ``limits'' of finite ones.
\end{itemize}

\end{itemize}
Our example will arise from a topologizing of non-well-founded set theory \cite{FH,Acz}.
We will end up with something logically weaker, but  computationally more meaningful --- and perhaps with some logical interest in its own right.
Along the way, we shall touch on numerous points in denotational semantics, concurrency theory, modal logic and set theory.

\paragraph{Acknowledgements} We will draw on ideas from many sources, notably: Peter Aczel on non-well-founded set theory \cite{Acz}; Nivat and de Bakker and Zucker on denotational models of processes based on ultrametrics \cite{Niv,deBZ}; Milner, Hennessy, Park, Bergsta and Klop on process algebra and bisimulation \cite{HM,Par,BK}; and Smyth's topological perspective on computation \cite{Smy}.
We shall also draw extensively on our own work on Domain Theory in Logical Form, A Domain Equation for Bisimulation, and Total vs. Partial Objects in Denotational Semantics \cite{Abr85,Abr87,Abr91,Abr91a}. The reader in search of more details is directed to these papers.

\section{First approach: finitary sets as limits of finite ones}
We begin with the \emph{hereditarily finite sets} $\HF$. These have played a role in logic as a ``roomier'' and more structured alternative to the natural numbers $\Nat$.
They can be defined inductively by
\begin{equation}
\frac{x_1 \in \HF , \ldots , x_n \in \HF}{\{ x_1 , \ldots , x_n\} \in \HF}
\end{equation}
(The base case is $n=0$, which gives $\varnothing \in \HF$.)
More formally, we can write the inductive definition
\[ 
\begin{array}{ccc}
  V_0 &   = &   \varnothing \\
  V_{k+1} &   = &   \power{V_k}\\
  V_{\omega} &   = &   \bigcup_{k \in \omega} V_k
\end{array}
\]
and define $\HF = V_{\omega}$.

This definition relies implicitly on the fact that the full powerset construction, applied to a finite set, can only yield finite subsets.
A more conceptually pleasing definition is to use the finite powerset constructor $\powerf{\cdot}$ explicitly. We can then  write a fixpoint equation for $\HF$:
\begin{equation}
\label{HFd3}
\HF = \mu X. \, \powerf{X} = \bigcup_{k \in \omega} \powerfk{\varnothing}
\end{equation}
by the least fixpoint theorem, since $\powerf{\cdot}$, unlike $\power{\cdot}$, is continuous.

This starts another train of thought. The finite powerset construction $\powerf{\cdot}$ builds \emph{free semilattices} over sets --- \ie $(\Powerf , \{\cdot\}, \bigcup )$ is a monad on $\Set$ whose algebras are the semilattices \cite{Mac}. Thus Definition~(\ref{HFd3}) suggest an algebraic description of $\HF$.

\subsection{$\HF$ as the free process algebra with one action}

We  consider  the following  algebraic theory:
\subsubsection*{A fragment of CCS \cite{Mil}}
The signature comprises a binary operation $+$, a unary \emph{prefixing operator} $e\cdot$, and a constant $0$.
The equations are:
\begin{eqnarray}
x + 0 & = & x \\
x + y & = & y + x \\
x + (y + z) & = & (x+y) + z \\
x + x & = & x
\end{eqnarray}
which say that an algebra $(S, +, 0)$ is a semilattice  (equivalently, an idempotent, commutative monoid). There are no axioms for the prefixing operator.
We consider the free algebra with no generators.
(Almost equivalently, we can consider the algebra BA of Bergstra and Klop \cite{BK} over one generator.)


The intention  is indicated by the following semantics:
\begin{eqnarray*}
\lsem 0 \rsem & = & \varnothing \\
\lsem P + Q \rsem & = & \lsem P \rsem \cup \lsem Q \rsem \\
\lsem e \cdot P \rsem & = & \{ \lsem P \rsem \} .
\end{eqnarray*}
Thus iterations of prefixing allows arbitrary (finite) levels of nesting of sets, while within each level the semilattice axioms give the structure of a finite powerset.\footnote{Note that, if we built \emph{complete} semilattices at each level, and iterated the construction through all ordinals, we would (disregarding  set-theoretic subtleties) be building a full set-theoretic universe $V$ as the initial algebra for this theory! Subsequently, Joyal and Moerdijk have developed a sophisticated  form of ``algebraic set theory'' in a general categorical setting, which in essence uses this approach to the description of set-theoretic universes \cite{JM}.}

We can therefore define a membership relation on this free algebra:
\[ [P] \in [Q] \;\; \Longleftrightarrow \;\; \exists R. \, Q = e\cdot P + R . \]
Note that this is equivalently expressed as the usual \emph{transition relation} on processes:
\[  [P] \in [Q] \;\; \Longleftrightarrow \;\;  Q \stackrel{e}{\longrightarrow} P  \]
This leads to a more computational, process-like description of hereditarily finite sets, as \emph{finite, rooted trees}.
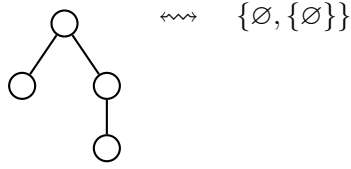
\begin{example}
\[ \pstree[levelsep=5ex]{\Tcircle{\;}}{\Tcircle{\;}
\pstree{\Tcircle{\;}}{\Tcircle{\;}}}
\quad \leftrightsquigarrow \quad \{ \vn , \{ \vn \}\}
\]
\end{example}

We can characterize equality on process terms or finite rooted trees directly in terms of the transition relation: this is the now-classical notion of \emph{bisimulation} \cite{Par,HM,Mil}:
\begin{eqnarray}
T_1 \sim T_2  & \Longleftrightarrow & T_1 \rightarrow T_{1}' \;\; \Rightarrow \;\; \exists T_{2}' . \, T_2 \rightarrow T_{2}' \; \wedge \; T_{1}'  \sim T_{2}' \\
& \wedge & T_2 \rightarrow T_{2}' \;\; \Rightarrow \;\; \exists T_{1}' . \, T_1 \rightarrow T_{1}' \; \wedge \; T_{1}'  \sim T_{2}' .
\end{eqnarray}
This recursively simulates extensionality:
\begin{eqnarray}
\label{exdef}
x = y & \Longleftrightarrow & \forall z. \, z \in x \; \leftrightarrow \; z \in y \\
x \sim y & \Longleftrightarrow &  \forall z. \, z \in x \;\; \Rightarrow \;\; \exists w. \, w \in y  \; \wedge z \sim w \\
& & \wedge \;\;  z \in y \;\; \Rightarrow \;\; \exists w. \, w \in x  \; \wedge z \sim w .
\end{eqnarray}
This is indeed reminiscent of the way that extensionality is imposed in various constructions of models of set theory (e.g. Boolean-valued models in forcing), and pointed the way to Aczel's work on non-well-founded set theory.

The recursion in the above \emph{coinductive definition} can,  in this simple case of finite trees, be unwound inductively as follows  \cite{HM}:
\begin{eqnarray*}
T_1  \sim_0 T_2 & \equiv & \mbox{true} \\
T_1  \sim_{k+1} T_2 & \equiv & T_1 \rightarrow T_{1}' \; \Rightarrow \; \exists T_{2}' . \, T_2 \rightarrow T_{2}' \; \wedge \; T_{1}'  \sim_k T_{2}' \\
& & \wedge \\
& & T_2 \rightarrow T_{2}' \; \Rightarrow \; \exists T_{1}' . \, T_1 \rightarrow T_{1}' \; \wedge \; T_{1}'  \sim_k T_{2}'  \\
T_1 \sim T_2  & \equiv & \forall k. \, T_1 \sim_k T_2 .
\end{eqnarray*}

\subsection{A metric  on $\HF$}
Since every set is the directed union of its finite subsets, we need a more refined notion of limit to filter out computationally unrealistic sets. Thus topology begins to enter the picture.

The rooted tree representation of $\HF$, and its associated bisimulation equivalence. gives rise to a computationally meaningful notion of distance between hereditarily finite sets:
\begin{center}
\fbox{The more work you have to do to distinguish between the sets, the closer they are.}
\end{center}
We equate `work' with the depth to which we have to probe the trees to discover a difference which distinguishes them as representations of sets. This leads to the following definition:
\[ d(S, T) \;\; = \;\; \begin{cases} 0, & \qquad S \sim T \\
2^{-k}, \;\; \text{least $k$ such that $S \not\sim_k T$} & \text{otherwise}
\end{cases}
\]
\begin{example}
\[ d(\varnothing , \{ \varnothing \} ) = 1/2, \qquad \qquad d(\{ \varnothing \}, \{  \varnothing , \{ \varnothing\}\}) = 1/4 .
\]
\end{example}
We can give an  inductive definition of this distance function, defined directly on hereditarily finite sets, as follows:
\[
\begin{array}{lcr}
(1)  &   d(S, T) = 0 &   \text{If  $S=T$} \\
(2.1)   &   d(\varnothing , T) = d(S, \varnothing ) = 1/2 & \text{If $S \neq \varnothing \neq T$}  \\
(2.2) & d(S, T) = 1/2 \max (\sup_{s \in S} \inf_{t \in T} d(s, t), \; \sup_{t \in T} \inf_{s \in S} d(s, t)) & \text{If $S \neq \varnothing \neq T$}
\end{array}
\]
Note how, in the inductive case (2.2), the standard \emph{Hausdorff metric} \cite{Dug} appears as a ``minimaxing'' calculation of $\not\sim_k$. We can also view the Hausdorff metric as the interpretation of the definition of extensional equality (\ref{exdef}) in ``real-valued logic''  \cite{Law}.

This definition does indeed yield a metric, and in fact an ultra-metric:
\[ d(S, T) \leq \max (d(S,U), d(U, T)) . \]

Now that we have an (ultra-)metric space $(\HF , d)$, we can look at Cauchy sequences to see which limits should arise.
\begin{example}
The sequence $\vn , \{\vn\}, \{\{\vn\}\}, \{\{\{\vn\}\}\}, \ldots $
can be visualized as follows, with successive distances indicated:
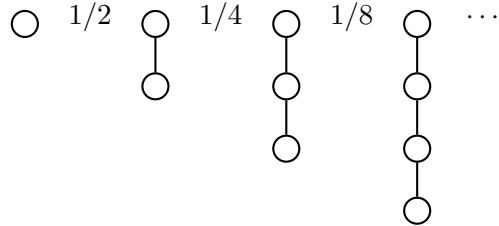
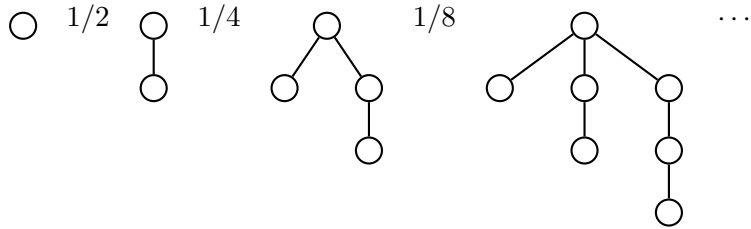
\label{dex1}
\[
\Tcircle{\;}
\quad 1/2 \quad
\pstree[levelsep=5ex]{\Tcircle{\;}}{\Tcircle{\;}}
\quad 1/4 \quad
\pstree[levelsep=5ex]{\Tcircle{\;}}{\pstree{\Tcircle{\;}}{\Tcircle{\;}}}
\quad 1/8 \quad
\pstree[levelsep=5ex]{\Tcircle{\;}}{\pstree{\Tcircle{\;}}{\pstree{\Tcircle{\;}}{\Tcircle{\;}}}}
\quad \cdots
\]
\end{example}
\begin{example}
\label{dex2}
\[
\Tcircle{\;}
\quad 1/2 \quad
\pstree[levelsep=5ex]{\Tcircle{\;}}{\Tcircle{\;}}
\quad 1/4 \quad
\pstree[levelsep=5ex]{\Tcircle{\;}}{\Tcircle{\;}
\pstree{\Tcircle{\;}}{\Tcircle{\;}}}
\quad 1/8 \quad
\pstree[levelsep=5ex]{\Tcircle{\;}}{\Tcircle{\;}
\pstree{\Tcircle{\;}}{\Tcircle{\;}}
\pstree{\Tcircle{\;}}{\pstree{\Tcircle{\;}}{\Tcircle{\;}}}}
\quad \cdots
\]
The corresponding sequence of sets is
\[ \vn ,\;\; \{ \vn \} , \;\; \{ \vn , \{ \vn \}\} , \;\; \{ \vn , \{ \vn \} , \;\; \{ \{ \vn \} \} \} , \ldots \]
\end{example}
We now finally reach our first definition of the \emph{universe of finitary sets} $\Fin$:
\begin{definition}
We define $\Fin$ to be the metric completion of $(\HF, d)$; \ie (equivalence classes of) Cauchy sequences in $\HF$.
Membership is defined by:
\[ [S_n ] \in [T_n ] \;\; \equiv \;\; \forall n. \, \exists m.\, S_n \in T_m . \]
\end{definition}
Note that in Example~(\ref{dex1}), these definitions yield $S \in S$! Thus the process of metric completion automatically gives rise to non-well-founded sets.

\section{Interlude: Domain Equations}
We consider domain equations $X \cong F(X)$.
An \emph{interpretation} of such an equation is provided by specifying a category $\CC$, and an endofunctor $F : \CC \longrightarrow \CC$. We are generally interested in an \emph{extremal solution} of such an equation: either an \emph{initial algebra} $\alpha : FA \rightarrow A$, or a \emph{final coalgebra} $\beta : A \rightarrow FA$. (The Lambek lemma \cite{Lam} guarantees that initiality or finality does indeed imply that the arrow is an isomorphism).
These concepts generalize the lattice-theoretic notions of least and greatest fixpoint.
In most cases of interest, initial algebras can be constructed as colimits:
\[ \lim_{\rightarrow} ( \Init \rightarrow F\Init \rightarrow F^2\Init \rightarrow \cdots ) \]
generalizing the construction of the least fixpoint as $\bigvee_k F^k \bot$, while final coalgebras can be constructed as limits:
\[ \lim_{\leftarrow} ( \Term \leftarrow F\Term \leftarrow F^2\Term \leftarrow \cdots ) \]
generalizing the construction of the greatest fixpoint as $\bigwedge_k F^k \top$.
(In the domain theoretic case, the \emph{limit-colimit coincidence} \cite{Sco70} means that the two constructions coincide, and we obtain \emph{both} an initial algebra \emph{and} a final coalgebra.)
For the finitary case we are considering, the functors will be $\omega$-continuous in the appropriate sense, and the limit or colimit can be taken with respect to the $\omega$-chain of finite iterations.

In \cite{Abr85}, the view was taken that 
\begin{center}
\fbox{Domain equations can, and should, be viewed schematologically.}
\end{center}
This means that a given language for describing functors can be interpreted in different categories, and the solutions compared. A general result of this kind is presented in \cite{Abr85}.
We shall follow this point of view in our treatment of $\Fin$.

\section{Second description of $\Fin$}
We will now characterize the space of finitary sets $\Fin$ as the solution of a domain equation. We take as our ambient category $\CUM$, the category of compact ultrametric spaces and continuous maps.
Our functor $F : \CUM \longrightarrow \CUM$ is
\[ F(X) = \Term + \powerH{X} \]
where $\powerH{X}$ is the set of all non-empty closed subsets of $X$ (note that `closed' and `compact' are equivalent in this context), with the Hausdorff metric:
\[ d(S, T) =  \max (\sup_{s \in S} \inf_{t \in T} d(s, t), \; \sup_{t \in T} \inf_{s \in S} d(s, t)) . \]
The $\Term + \cdots$ term is to code the empty set. We apply a contraction factor (for convenience, $1/2$) in taking the disjoint union $X + Y$, so that the empty set is a distance of at least $1/2$ from any non-empty set, and we shrink the Hausdorff metric by $1/2$ in each recursive iteration.

We now provide our second characterization of $\Fin$:
\begin{proposition}
$\Fin$ is the final coalgebra of the functor $F$. Moreover, $F$ is cocontinuous, so $\Fin$ is constructed as the limit of the $\omega^{\text{op}}$-chain
\[ \lim_{\leftarrow} ( \Init \leftarrow F\Init \leftarrow F^2\Init \leftarrow \cdots ) \]
\end{proposition}
We give a picture of the first few terms of the construction:
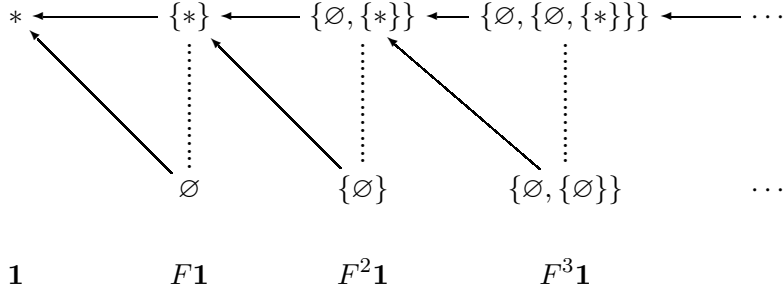
\begin{diagram}
\ast & \lTo & \{ \ast\} & \lTo & \{\vn ,\{\ast\}\} & \lTo & \{\vn ,\{\vn , \{\ast\}\}\} & \lTo & \cdots \\
& \luTo & \uDotline & \luTo & \uDotline & \luTo & \uDotline & & \\
& & \vn &  & \{ \vn \} & & \{ \vn , \{ \vn\}\} &  &   \cdots \\
\Term & & F\Term & & F^2\Term & & F^3\Term
\end{diagram}

This characterization analyzes our ad hoc inductive definition into the iterated application of a general construction; similarly, the metric completion of the finite levels arises systematically from the general notion of limit used to construct the final coalgebra.

We therefore articulate the following principle:
\begin{center}
\fbox{A domain equation yields more information than an ad hoc construction.}
\end{center}

The metric structure on $\Fin$ has some  significance, e.g. we can apply the Banach fixpoint theorem to deduce that the equation
$x = \{ x\}$
has a unique solution, since the map $x \mapsto \{x\}$ is contractive, and similarly for
\[ x \; = \; \{\vn\} \cup \{\{y\} \mid y \in x\} \]
(yielding our previous examples~(\ref{dex1}) and (\ref{dex2})).
More generally, we get the existence of unique solutions for \emph{guarded equations} (those in which the recursion variables appear under the scope of set-forming braces) \cite{Acz,BM}.
However, for many purposes the topological structure suffices --- and in any event, the precise definition of the metric is irrelevant to the structure. This leads naturally to the following \\
\textbf{Question:} Which metric topologies arise in $\CUM$?

\begin{proposition}
The category $\CUM$ is equivalent to the category $\Stone$ of second-countable Stone spaces.
\end{proposition}
\begin{proof}
In one direction, note that an open  ball in an ultrametric is a closed set. Indeed, if $z \in B(x;\epsilon )$ and $d(x,y) \geq \epsilon$, then
\[ \epsilon \leq d(x,y) \leq \max (d(x,z), d(z,y)) \]
so $d(z,y) \geq \epsilon$, and $z \not\in B(y;\epsilon )$. Thus $y \not\in \overline{B(x;\epsilon )}$.
This implies that the open balls form a clopen base in the metric topology.
Since $M$ is compact, it is second-countable.

For the converse, if a Stone space $S$ is second countable, with dual Boolean algebra
$B = \{ b_{n} \mid n \in \Nat \}$, define
\[ d(x, y) = \begin{cases}
0 & \quad x=y \\
2^{-n}, \quad  \text{least $n$ such that $x \in b_{n} \Leftrightarrow y \not\in b_{n}$} & \text{otherwise.}
\end{cases}
\]
\end{proof}

This equivalence leads us to our next characterization of $\Fin$.

\section{Third description of $\Fin$}
We now characterize $\Fin$, \textit{qua} topological space, as the solution of a domain equation in $\Stone$. To do this, we need to answer the following question:
\begin{center}
\fbox{What is the topological construction analogous to the Hausdorff metric powerspace $\powerh$?}
\end{center}
The answer is provided by the \emph{Vietoris construction}  $\powerv$ \cite{Joh}.
Although it can be defined much more generally, we shall  view $\powerv$ as a functor on $\Stone$. Given a Stone space $S$, $\powerv (S)$ is the set of all compact (which since $S$ is compact Hausdorff, is equivalent to closed) subsets of $S$, with topology generated by
\begin{eqnarray}
\Box U & = & \{ C \mid C \subseteq U \} \\
\Diamond U & = & \{ C \mid C \cap U \neq \vn \} 
\end{eqnarray}
where $U$ ranges over the open sets of $S$. We can read $\Box U$ as the set of all $C$ such that $C$ \emph{must} satisfy $U$, and $\Diamond U$ as the set of $C$ such that $C$ \emph{may} satisfy $U$. The allusion to modal logic  notation is thus deliberate, and we shall shortly see a connection to standard modal notions.
Note that in our definition of the Vietoris powerspace, the empty set \emph{is} included. 

We now compare  the Vietoris topology, and the metric topology arising from the Hausdorff powerspace metric. Given a metric space $M$, we write $\MT{M}$ for the topological space arising by taking the metric topology on $M$.
\begin{proposition}
For any $M$ in $\CUM$: $\MT{\Term + \powerH{M}} = \powerv (\MT{M})$.
\end{proposition}
This result is in fact true in much greater generality, but the above is sufficient for our purposes.

Now we obtain the following description of $\Fin$.
\begin{proposition}
The space of finitary sets in its metric topology, $\MT{\Fin}$, is the final coalgebra of the Vietoris functor on $\Stone$.
\end{proposition}
This is an instance of a general result comparing the solutions of a class of domain equations in $\CUM$ with the corresponding solutions in $\Stone$ \cite{Abr85}. (The class covers many of the ultrametric process models.)
This description will be very useful in  exploring the structure of $\Fin$ as a set-theoretic universe.

Since $\Fin$ is a Stone space, it has a dual Boolean algebra. This can be derived systematically --- as a logic --- from the domain equation used to describe $\Fin$. This illustrates the general programme for relating denotational semantics and program logics developed under the heading of `Domain theory in logical form' \cite{Abr87,Abr91}. It will also lead us to our next description of $\Fin$.

\section{Fourth description of $\Fin$}
We now describe $\Fin$ as \emph{the Stone space of the free modal algebra (on no generators)}.
Here we take a modal algebra to be a Boolean algebra $B$ equipped with a unary operator $\Diamond$ satisfying the axioms
\[ \text{(MA)} \quad \Diamond (a \vee b) = \Diamond a \vee \Diamond b \qquad \qquad \Diamond 0 = 0 . \]
This is the algebraic variety corresponding to the minimal normal modal logic $\KK$ \cite{BRV}.
The Boolean algebra is equipped with a constant  $0$, so the free algebra over no generators can be non-trivial. We shall show that it is indeed non-trivial.

To derive this characterization systematically, we recall that the Vietoris construction can be described logically (or localically \cite{Joh}) as an operation on \emph{theories}. For the coherent case \cite{Joh}, $V(L)$, for a distributive lattice $L$, is the distributive lattice generated by $\Box a$, $\Diamond a$, ($a \in L$), subject to the axioms:
\begin{eqnarray}
\label{VLax1}
\Box (a \wedge b) = \Box a \wedge \Box b &   \qquad &   \Diamond (a \vee b) = \Diamond a \vee \Diamond b \\
\label{VLax2}
 \Box 1  = 1  &   &   \Diamond 0 = 0 \\
\label{VLax3}
 \Box (a \vee b) \leq \Box a \vee \Diamond b &   &   \Diamond (a \wedge b) \geq \Diamond a \wedge \Box b .
\end{eqnarray}
In the boolean case, where we have a classical negation, $\Box$ and $\Diamond$ are inter-definable (e.g. $\Box a = \neg \Diamond \neg a$), and the axiomatization simplifies to (MA). To see this, note firstly that (MA) implies that $\Diamond$ is monotone.  Now from Boolean algebra we derive $a \leq \neg b \vee (a \wedge b)$, from which by monotonicity we obtain
\[ \Diamond a \;\; \leq \;\; \Diamond (\neg b \vee (a \wedge b)) \;\;  = \;\; \Diamond \neg b \; \vee \Diamond (a \wedge b) , \]
and hence $ \Diamond a \wedge \neg \Diamond \neg b \leq \Diamond (a \wedge b)$.
The rest of  (\ref{VLax1})--\ref{VLax3}) follows by duality.
Thus we obtain a construction $\MA{B}$, which when applied to a Boolean algebra $B$ constructs a new Boolean algebra with generators $\Diamond a$, $a \in B$, subject to the Boolean algebra axioms plus (MA).

Now  we can iterate this construction to get the initial solution of
$\BB = \MA{\BB}$ in $\Bool$, the category of Boolean algebras. This is constructed by taking a colimit of the finite iterates $\MAk{\Two}$, starting from the 2-element Boolean algebra $\Two$. This colimit can be constructed concretely as a union, essentially by taking the Lindenbaum algebra of the propositional theory which is inductively generated by these iterates. This propositional theory is the standard modal system $\KK$---but with no propositional atoms. (There \emph{are} constants for \textsf{true} and/or \textsf{false}.) Thus another role for domain equations is revealed, as systematizing the inductive definition of the formulas and inference rules of a logic. A detailed account of how this works in a directly analogous (and rather more complex) case is given in \cite{Abr87,Abr91}.

To see how hereditarily finite sets can be completely characterized by modal formulas (the ``master formula'' of the set; \textit{cf.} \cite{HM}), we define:
\begin{eqnarray}
\Form{\vn} & = & \Box 0 \qquad ( = \neg \Diamond 1) \\
\Form{\{x_{1} , \ldots , x_{n} \}} & = & \Box \bigvee_{i=1}^{n} \Form{x_{i}} \; \wedge\; \bigwedge_{i=1}^{n}  \Diamond \Form{x_{i}} .
\end{eqnarray}
The link between $\BB$ and $\Fin$ is given by \emph{Stone duality} \cite{Joh}:
\begin{proposition}

\begin{enumerate}
\item $\BB$ is isomorphic to the Boolean algebra of clopen subsets of $\Fin$, with the modal operator $\Diamond$ defined by
\[ \Diamond U = \{ S \in \Fin \mid S \subseteq U \} \]
(recall that $\Fin \cong \powerV{\Fin}$).
\item $\Fin$ is isomorphic to Spec $\BB$, the space of ultrafilters over $\BB$, \ie \emph{models} of $\BB$: $f^{-1}(1)$, for Boolean homomorphisms $f : \BB \longrightarrow \Two$. This space is topologised by:
\[ U_{b} = \{ x \in \text{Spec} \; \BB \mid b \in x \} \qquad (b \in \BB )  . \]
\end{enumerate}
\end{proposition}
Again, this is an instance of very general results \cite{Abr91}.

\section{Fifth description of $\Fin$}
Our final characterization of $\Fin$ is as the subspace of maximal elements of a \emph{domain} \cite{Sco70,AJ}. Once again, this characterization will arise systematically, by comparing our previous description of $\Fin$ as the solution of a domain equation in $\Stone$ with the solution of a corresponding equation in a category of domains $\Dom$.
A convenient choice for $\Dom$ for our purposes is $\SFP$ \cite{Plo}, the countably-based bifinite domains \cite{AJ}.

We define a domain $D$ as the solution (both initial algebra and final coalgebra) of the equation
\begin{equation}
\label{domdef}
D \;\; = \;\; \powerPO{D} = \;\; \lift{\Term} \oplus \powerP{D} . 
\end{equation}
Here $\powerP{\cdot}$ is the Plotkin or convex powerdomain \cite{Plo,AJ}. The additional term in our domain equation codes in the empty set, by a ``semi-coalesced sum'':
\[
\lift{\Term} \oplus D \;\; = \;\;
\pstree[treemode=U]{\TR{}}{\TR{\ast} 
\pstree[linestyle=none,arrows=-,levelsep=1ex]{\Tfan}{\TR{D}}}
\] 
A similar but more general domain equation, which allows for an arbitrary  set of possible actions, is used in \cite{Abr87,Abr91a} to give a denotational semantics for process calculi, which is fully abstract with respect to strong \emph{partial} bisimulation.

The Plotkin powerdomain on $\SFP$ (or more generally, Lawson-compact) domains, in fact coincides with the Vietoris construction  on those domains seen as topological spaces with the Scott topology  \cite{Smy,Abr87,AJ}.

Given a domain  $D$, we write $\Max{D}$ for the subspace of maximal elements, viewed as a topological space. We can now state our final characterization of $\Fin$.
\begin{proposition}
$\Fin \cong \Max{D}$, where $D$ is the solution of the domain equation (\ref{domdef}).
\end{proposition}
Again, this is an instance of a general result, relating solutions of a class of domain equations in $\Dom$ with those in $\Stone$ \cite{Abr85}.

We note that $D$ has ``partial sets''.
\begin{example}
\[ \pstree[levelsep=5ex]{\Tcircle{\bot}}{\Tcircle{\;\;}
\pstree{\Tcircle{\bot}}{\Tcircle{\;\;}}}
\quad \leftrightsquigarrow \quad \{ \bot , \vn , \{ \bot , \vn \} \} 
\]
\end{example}
The ordering on $D$ is the Egli-Milner ordering \cite{Plo}:
\[ S \sqsubseteq T \;\; \equiv \;\;  (\forall x \in S. \, \exists y \in T. \, x \sqsubseteq y) \; \wedge \;
 (\forall y \in T. \, \exists x \in S. \, x \sqsubseteq y) .
 \]
 This can be described in terms of ``partial bisimulation''; see e.g. \cite{Abr91a}.
 
The space of partial sets has been explored further in \cite{OMM,ABH}.
 
 \section{Set Theory in $\Fin$}
We now turn to considering $\Fin$ as a set-theoretic universe, \ie as a structure $(\Fin , {\in}, {=})$. Since we have the isomorphism
 \begin{diagram}
 \Fin && \pile{\rTo^{\mathsf{unfold}} \\ \cong \\ \lTo_{\mathsf{fold}}}&& \powerV{\Fin}
 \end{diagram}
 we can define
 \[ S \in T \;\; \equiv \;\; S \in \mathsf{unfold} (T) . \]
(We shall often elide uses of $\mathsf{fold}$ and $\mathsf{unfold}$.)

How much of ordinary (ZFC) set theory does $\Fin$ satisfy? It certainly cannot be a model of full ZFC, since $\Fin$ satisfies the ultimate anti-foundation axiom:
 \begin{center}
\fbox{$V$ is a set}
\end{center}
where $V = \{ x \mid \mathsf{true}\} = \Fin$. So something has to give, or Russell's paradox would apply.

We begin by discussing the axioms which do straightforwardly apply. Our knowledge of the structure of the domain equation used to construct $\Fin$ makes these properties fall out very simply.
\begin{itemize}
\item Extensionality:
\[ [\forall z. \, z \in x \; \leftrightarrow \; z \in y] \; \rightarrow \; x = y . \]
This follows immediately from the injectivity of $\mathsf{unfold}$.
\item The following closure conditions:
\begin{itemize}
\item $\vn \in \Fin$
\item $x \in \Fin \; \Rightarrow \; \{x\} \in \Fin$
\item $x, y \in \Fin \; \Rightarrow \; x \cup y \in \Fin$
\item $x \in \Fin \; \Rightarrow \; \bigcup x = \{ z \mid \exists y \in x. \, z \in y\} \in \Fin$
\end{itemize}
all hold, since $(\powerV{\Fin}, {\cup}, \vn )$ is a semilattice (the free topological semilattice generated by $\Fin$), and $(\powerv, \{ \cdot \}, {\bigcup})$ is a monad  \cite{Joh}.
\item A suitable axiom of Infinity holds, since e.g.
\[ x = \{\vn\} \cup \{\{y\} \mid y \in x \} \]
has a (unique) solution in $\Fin$.
\item The Powerset axiom holds, in the form:
\[ x \in \Fin \;\; \Rightarrow \;\;  \powerV{x} \in \Fin . \]
Indeed $\mathsf{unfold}(x) \subseteq \Fin$; applying $\powerv$ functorially to the inclusion map we get 
\[ \powerV{\mathsf{unfold}(x)} \hookrightarrow \powerV{\Fin} . \]
The image of a compact space is compact, and hence $\powerv (\mathsf{unfold}(x)) \in \powerv (\powerv (\Fin ))$, and $\mathsf{fold}^{2}(\powerV{\mathsf{unfold}(x)}) \in \Fin$.
\end{itemize}
Clearly, we cannot have full (classical) separation, since otherwise we could derive an inconsistency from $V \in V$.
We \emph{do}, however, have \emph{continuous} versions of Separation, Replacement, and Choice.

\paragraph{Separation}
Separation says that $\{ y \in x \mid \phi (y) \}$ is a set. Consider the chacteristic function of the predicate $\phi$, $f_{\phi} : \Fin \rightarrow \Two$. If $f_{\phi}$ is \emph{continuous} ($\Two$ taken as a discrete space), then 
\[ \{ y \in x \mid f_{\phi (y)} = \mathsf{true} \} \in \Fin , \]
since this is the inverse image of a closed set by a continuous map, hence closed, hence compact, hence in $\Fin$.

\paragraph{Replacement}
If $f : x \rightarrow \Fin$ is continuous, for $x \in \Fin$:
\[ \{ f(y) \mid y \in x \} \in \Fin . \]
Since $x$ is compact, the image $f(x) \subseteq \Fin$ is compact, hence $f(x) \in \powerV{\Fin}$, and $\mathsf{fold}(f(x)) \in \Fin$.

\paragraph{Choice}
If $f :  x \longrightarrow \powerV{\Fin} \setminus \{ \vn\}$ is continuous, for $x \in \Fin$, then for some $y \in \Fin$:
\[ \forall z \in x. \, \exists w \in f(z). \, w \in y \;\; \wedge \;\; \forall w \in y. \, \exists z \in x. \, w \in f(z) . \]
This arises from a standard topological result about selection functions \cite{Mic}: given continuous $f : x \longrightarrow \powerV{\Fin} \setminus \{ \vn\}$, there is a continuous selection function $g : x \rightarrow \Fin$ with $g(z) \in f(z)$, for all $z \in x$. We can then take
\[ y = \{ g(z) \mid z \in x \} . \]

In order to formulate a logical theory for this set-theoretic universe, we need to impose syntactic conditions on formulas to ensure that they give rise to continuous functions. These conditions will, inevitably, involve restrictions on the use of negation. They can then be used to formulate appropriate versions of the Separation, Replacement and Choice axioms.

For more extensive investigations in this area, see e.g. \cite{FH1}.


\begin{thebibliography}{Mil89}

\bibitem{Abr85}
S.~Abramsky.
\textsl{Total vs. Partial Objects in Denotational Semantics}.
Unpublished lecture, given at Workshop on Category Theory and Computer Programming, Guildford,
1985.

\bibitem{Abr87}
S.~Abramsky.
\textsl{Domain Theory and the Logic of Observable Properties}.
Ph.D. thesis, University of London, 1987.
Available at \texttt{http://web.comlab.ox.ac.uk/oucl/work/samson.abramsky/}.

\bibitem{Abr91}
S.~Abramsky.
Domain theory in logical form.
\textsl{Annals of Pure and Applied Logic} 51, 1--77, 1991.

\bibitem{Abr91a}
S. Abramsky. 
A Domain Equation for Bisimulation.
\textsl{ Information
and Computation}, 92(2), 161--218, 1991. 

\bibitem{AJ}
S. Abramsky and A. Jung.
Domain Theory.
\textsl{Handbook of Logic in Computer Science}, edited by
S. Abramsky, D. M. Gabbay and T. S. E. Maibaum, Oxford University Press,
Vol. 3, 1--168, 1994.

\bibitem{Han}
\textsl{The Handbook of Logic in Computer Science},
edited by S. Abramsky, D. M. Gabbay and T. S. E. Maibaum.
Oxford University Press. 
Volume 1:  {\slshape Background: Mathematical Structures} and Volume 2:
{\slshape Background: Computational Structures},  published in 1992.
Volume 3:  {\slshape Semantic Structures} and Volume 4:  {\slshape Semantic Modelling},
published in 1995. 
Volume 5:  {\slshape Logic and Algebraic Methods}, published in 2000.


\bibitem{Acz}
Peter Aczel.
\textsl{Non-Well-Founded Sets}.
CSLI Lecture Notes Vol. 14.
Stanford University. 1988.

\bibitem{ABH}
F. Alessi, P. Baldan and F. Honsell.
Partializing Stone Spaces Using SFP Domains.
\textsl{Proceedings of CAAP `97}, Springer Lecture Notes in Computer Science Vol. 1158:478--489, 1997.

\bibitem{deBZ}
J. de Bakker and J. Zucker.
Processes and the denotational semantics of concurrency.
\textsl{Information and Control} 54:70--120, 1982.

\bibitem{BM}
J. Barwise and L. Moss.
\textsl{Vicious Circles}.
CSLI Publications, 1996.

\bibitem{BK}
J. Bergstra and J.-W. Klop.
Process algebra for synchronous communication.
\textsl{Information and Control} 60:109--137, 1984.

\bibitem{BRV}
P. Blackburn, M. de Rijke and Y. Venema.
\textsl{Modal Logic}.
Cambridge University Press 2001.

\bibitem{Dug}
J. Dugundji.
\textsl{Topology}.
Allyn and Bacon 1966.

\bibitem{FH}
M. Forti and F. Honsell.
Set Theory with Free Construction Principles.
\textsl{Annali Scuola Normale Superiore --- Pisa Classe de Scienza 10:493--522. Serie IV}. 1983.

\bibitem{FH1}
M. Forti and F. Honsell.
A General Construction of Hyperuniverses.
\textsl{Theoretical Computer Science} 156:203--215, 1996.

\bibitem{HPL}
D. Gabbay and F.Guenthner.
\textsl{Handbook of Philosophical Logic Volumes I--IV}.
First Edition,
D. Reidel, 1983--89.

\bibitem{HM}
M. Hennessy and R. Milner.
Algebraic laws for non-determinism and concurrency.
\textsl{Journal of the ACM} 32(1):137--161, 1985.

\bibitem{Joh}
P. T. Johnstone.
\textsl{Stone Spaces}.
Cambridge University Press 1982.

\bibitem{JM}
A. Joyal and I. Moerdijk.
\textsl{Algebraic Set Theory}.
Cambridge University Press, 1995.

\bibitem{Lam}
J. Lambek.
Subequalisers.
\textsl{Canad. Math. Bull.} 13:337--349, 1970.

\bibitem{Law}
F. W. Lawvere.
Metric spaces, generalized logic and closed categories.
\textsl{Rendiconti Seminario Matematico e Fisico di Milano XLIII}, Pavia 1973.

\bibitem{Mac}
S. Mac Lane.
\textsl{Categories for the Working Mathematician}.
Second Edition.
Springer-Verlag 1998.

\bibitem{Mic}
E. Michael.
Topologies on spaces of subsets.
\textsl{Trans. American Math. Soc.} 71, 152--182, 1951.

\bibitem{Mil}
R. Milner.
\textsl{A Calculus of Communicating Systems}.
Springer Lecture Notes in Computer Science Vol. 92, 1980.

\bibitem{OMM}
M. Mislove, F. Oles and L. Moss.
Non-well-founded sets modelled as ideal fixed points.
\textsl{Information and Computation} 93(1):16--54, 1991.

\bibitem{Niv}
M. Nivat.
Infinite words, infinite trees, infinite computations.
\textsl{Foundations of Computer Science III}.
Mathematics Centrum Tracts, 1979.

\bibitem{Par}
D. Park.
Concurrency and Automata on Infinite Sequences.
\textsl{Proceedings of the 5th GI Conference}.
Springer Lecture Notes in Computer Science, Vol. 104, 167--183, 1981.

\bibitem{Plo}
G. D. Plotkin.
A Powerdomain Construction.
\textsl{SIAM Journal on Computing}, 5:452--487, 1976.

\bibitem{Sco70}
D.~S.~Scott. \textsl{Outline of a Mathematical Theory
of Computation}. Technical Monograph PRG-2 OUCL, 1970.

\bibitem{Smy}
M. B. Smyth.
Powerdomains and predicate transformers: a topological view.
\textsl{Automata, Languages and Programming}, edited by J. Diaz,
Springer Lecture Notes in Computer Science Vol. 154, 662--675, 1983.

\end{thebibliography}
\end{document}